\documentclass[12pt,reqno]{article}
\usepackage[usenames]{color}
\usepackage{amssymb}
\usepackage{graphicx}
\usepackage{amscd}

\usepackage{graphics,amsmath,mathtools,amssymb,relsize}
\usepackage{bm}
\usepackage{amsthm}
\usepackage{amsfonts}
\usepackage{latexsym}
\usepackage{epsf}

\newtheorem{theorem}{Theorem}[section]
\newtheorem{corollary}[theorem]{Corollary}
\newtheorem{lemma}[theorem]{Lemma}
\newtheorem{proposition}[theorem]{Proposition}

\newtheorem{defin}[theorem]{Definition}
\newenvironment{definition}{\begin{defin}\normalfont\quad}{\end{defin}}
\newtheorem{examp}[theorem]{Example}

\newtheorem{rema}[theorem]{Remark}
\newtheorem{prob}[theorem]{Problem}

\numberwithin{equation}{section}

\newcommand{\bt}{\begin{thm}}
\newcommand{\et}{\end{thm}}
\newcommand{\bp}{\begin{proof}}
\newcommand{\ep}{\end{proof}}
\newcommand{\bprop}{\begin{prop}}
\newcommand{\eprop}{\end{prop}}
\newcommand{\bl}{\begin{lemma}}
\newcommand{\el}{\end{lemma}}
\newcommand{\bc}{\begin{corollary}}
\newcommand{\ec}{\end{corollary}}
\newcommand{\Z}{\mathbb{Z}}

\newcommand{\be}{\begin{enumerate}}
\newcommand{\ee}{\end{enumerate}}

\newcommand{\OMIT}[1]{}

\title{A generalization of Sch\"{o}nemann's theorem \\ via a graph theoretic method}
\author{Khodakhast Bibak \thanks{Department of Computer Science and Software Engineering, Miami University, Oxford, Ohio, 45056, USA. Email: {\tt bibakk@miamioh.edu}} \and Bruce M. Kapron \thanks{Department of Computer Science, University of Victoria, Victoria, BC, Canada V8W 3P6. Email: {\tt bmkapron@uvic.ca}} \and Venkatesh Srinivasan \thanks{Department of Computer Science, University of Victoria, Victoria, BC, Canada V8W 3P6. Email: {\tt srinivas@uvic.ca}}}

\begin{document}

\maketitle

\begin{abstract}
Recently, Grynkiewicz et al. [{\it Israel J. Math.} {\bf 193} (2013), 359--398], using tools from additive combinatorics and group theory, proved necessary and sufficient conditions under which the linear congruence $a_1x_1+\cdots +a_kx_k\equiv b \pmod{n}$, where $a_1,\ldots,a_k,b,n$ ($n\geq 1$) are arbitrary integers, has a solution $\langle x_1,\ldots,x_k \rangle \in \Z_{n}^k$ with all $x_i$ distinct. So, it would be an interesting problem to give an explicit formula for the number of such solutions. Quite surprisingly, this problem was first considered, in a special case, by Sch\"{o}nemann almost two centuries ago(!) but his result seems to have been forgotten. Sch\"{o}nemann [{\it J. Reine Angew. Math.} {\bf 1839} (1839), 231--243] proved an explicit formula for the number of such solutions when $b=0$, $n=p$ a prime, and $\sum_{i=1}^k a_i \equiv 0 \pmod{p}$ but $\sum_{i \in I} a_i \not\equiv 0 \pmod{p}$ for all $\emptyset \not= I\varsubsetneq \lbrace 1, \ldots, k\rbrace$. In this paper, we generalize Sch\"{o}nemann's theorem using a result on the number of solutions of linear congruences due to D. N. Lehmer and also a result on graph enumeration. This seems to be a rather uncommon method in the area; besides, our proof technique or its modifications may be useful for dealing with other cases of this problem (or even the general case) or other relevant problems.
\end{abstract}


\section{Introduction}\label{Sec 1}

Throughout the paper, we use $(a_1,\ldots,a_k)$ to denote the greatest common divisor (gcd) of the integers $a_1,\ldots,a_k$, and write $\langle a_1,\ldots,a_k\rangle$ for an ordered $k$-tuple of integers. Let $a_1,\ldots,a_k,b,n\in \Z$, $n\geq 1$. A linear congruence in $k$ unknowns $x_1,\ldots,x_k$ is of the form
\begin{align} \label{cong form}
a_1x_1+\cdots +a_kx_k\equiv b \pmod{n}.
\end{align}

By a solution of (\ref{cong form}), we mean an $\mathbf{x}=\langle x_1,\ldots,x_k \rangle \in \mathbb{Z}_n^k$ that satisfies (\ref{cong form}). The following result, proved by D. N. Lehmer \cite{LEH2}, gives the number of solutions of the above linear congruence:

\begin{proposition}\label{Prop: lin cong}
Let $a_1,\ldots,a_k,b,n\in \Z$, $n\geq 1$. The linear congruence $a_1x_1+\cdots +a_kx_k\equiv b \pmod{n}$ has a solution $\langle x_1,\ldots,x_k \rangle \in \Z_{n}^k$ if and only if $\ell \mid b$, where
$\ell=(a_1, \ldots, a_k, n)$. Furthermore, if this condition is satisfied, then there are $\ell n^{k-1}$ solutions.
\end{proposition}

Counting the number of solutions of the above congruence with some restrictions on the solutions is also a problem of great interest. As an important example, one can mention the restrictions $(x_i,n)=t_i$ ($1\leq i\leq k$), where $t_1,\ldots,t_k$ are given positive divisors of $n$. The number of solutions of the linear congruences with the above restrictions, which were called {\it restricted linear congruences} in \cite{BKSTT}, was first considered by Rademacher \cite{Rad1925} in 1925 and Brauer \cite{Bra1926} in 1926, in the special case of $a_i=t_i=1$ $(1\leq i \leq k)$, and they proved the following nice formula for the number $N_n(k,b)$ of such solutions:
\begin{align*}
N_n(k,b)= \frac{\varphi(n)^k}{n} \mathlarger{\prod}_{p\, \mid \, n, \,
p\, \mid\, b} \!\left(1-\frac{(-1)^{k-1}}{(p-1)^{k-1}}
\right)\mathlarger{\prod}_{p\, \mid\, n, \, p \, \nmid \, b}
\!\left(1-\frac{(-1)^k}{(p-1)^k}\right),
\end{align*}
where $\varphi(n)$ is Euler's totient function and the products are taken over all prime divisors $p$ of $n$. Since then, this problem has been studied, in several other special cases, in many papers (very recently, it was studied in its `most general case' in \cite{BKSTT}) and has found very interesting applications in number theory, combinatorics, geometry, computer science, cryptography etc; see \cite{BKS2, BKSTT3, BKSTT, BKSTT2, COH0, JAWILL} for a detailed discussion about this problem and a comprehensive list of references. Another restriction of potential interest is imposing the condition that all $x_i$ are {\it distinct}. Unlike the first problem, there seems to be very little published on the second problem. Recently, Grynkiewicz et al. \cite{GPP}, using tools from additive combinatorics and group theory, proved necessary and sufficient conditions under which the linear congruence $a_1x_1+\cdots +a_kx_k\equiv b \pmod{n}$, where $a_1,\ldots,a_k,b,n$ ($n\geq 1$) are arbitrary integers, has a solution $\langle x_1,\ldots,x_k \rangle \in \Z_{n}^k$ with all $x_i$ distinct; see also \cite{ADP, GPP} for connections to zero-sum theory and \cite{BKS7} for connections to coding theory. So, it would be an interesting problem to give an explicit formula for the number of such solutions. Quite surprisingly, this problem was first considered, in a special case, by Sch\"{o}nemann \cite{SCH} almost two centuries ago(!) but his result seems to have been forgotten. Sch\"{o}nemann \cite{SCH} proved the following result:

\begin{theorem} \label{Schonemann thm}
Let $p$ be a prime, $a_1,\ldots,a_k$ be arbitrary integers, and $\sum_{i=1}^k a_i \equiv 0 \pmod{p}$ but $\sum_{i \in I} a_i \not\equiv 0 \pmod{p}$ for all $\emptyset \not= I\varsubsetneq \lbrace 1, \ldots, k\rbrace$. The number $N_p(k)$ of solutions $\langle x_1,\ldots,x_k \rangle \in \Z_{p}^k$ of the linear congruence $a_1x_1+\cdots +a_kx_k\equiv 0 \pmod{p}$, with all $x_i$ distinct, is independent of the coefficients $a_1,\ldots,a_k$ and is equal to

$$
N_p(k)=(-1)^{k-1}(k-1)!(p-1)+(p-1)\cdots(p-k+1).
$$
\end{theorem}

In this paper, we generalize Sch\"{o}nemann's theorem using Proposition~\ref{Prop: lin cong} and a result on graph enumeration. This seems to be a rather uncommon method in the area; besides, our proof technique or its modifications may be useful for dealing with other cases of this problem (or even the general case) or other relevant problems. We state and prove our main result in the next section.

\section{Main Result}\label{Sec 2}

Our generalization of Sch\"{o}nemann's theorem is obtained via a graph theoretic method which may be also of independent interest. We need two formulas on graph enumeration (see Theorem~\ref{graph enum} below). These formulas are in terms of the {\it deformed exponential function} which is a special case of the {\it three variable Rogers-Ramanujan function} defined below. These functions have interesting applications in combinatorics, complex analysis, functional differential equations, and statistical mechanics (see \cite{ACH, KOSH, LANG, LIU, SCSO, SOKA} and the references therein).

\begin{definition}\label{Rog-Ram func}
The {\it three variable Rogers-Ramanujan function} is

$$
R(\alpha,\beta,q)=\sum_{m\geq0}\frac{\alpha^{m}\beta^{(_{2}^{m})}}{(1+q)(1+q+q^2)\cdots(1+q+\cdots+q^{m-1})}.
$$
Also, the {\it deformed exponential function} is

$$
F(\alpha,\beta)=R(\alpha,\beta,1)=\sum_{m\geq0}\frac{\alpha^{m}\beta^{(_{2}^{m})}}{m!}.
$$
\end{definition}

Let $g(c,e,k)$ be the number of simple graphs with $c$ connected components, $e$ edges, and $k$ vertices labeled $1,\ldots,k$, and $g'(e,k)$ be the number of simple {\it connected} graphs with $e$ edges and $k$ labeled vertices. Suppose that

$$
G(t,y,z)=\sum_{c,e,k}g(c,e,k)t^{c}y^{e}\frac{z^k}{k!},
$$
and

$$
CG(y,z)=\sum_{e,k}g'(e,k)y^{e}\frac{z^k}{k!}.
$$

\begin{theorem} {\rm (\cite{ACH, STAN2})} \label{graph enum}
The generating functions for counting simple graphs and simple connected graphs satisfy, respectively,

$$
G(t,y,z)=F(z,1+y)^t,
$$
and

$$
CG(y,z)=\log F(z,1+y),
$$
where $F$ is the deformed exponential function defined above.
\end{theorem}

Now, we are ready to state and prove our main result:

\begin{theorem} \label{Gener Schonemann thm}
Let $a_1,\ldots,a_k,b,n$ $(n\geq 1)$ be arbitrary integers, and $(\sum_{i \in I} a_i, n)=1$ for all $\emptyset \not= I\varsubsetneq \lbrace 1, \ldots, k\rbrace$. The number $N_n(b;a_1,\ldots,a_k)$ of solutions $\langle x_1,\ldots,x_k \rangle \in \Z_{n}^k$ of the linear congruence $a_1x_1+\cdots +a_kx_k\equiv b \pmod{n}$, with all $x_i$ distinct, is
\begin{align*}
& N_n(b;a_1,\ldots,a_k)\\
&=\begin{dcases}
(-1)^{k}(k-1)!+(n-1)\cdots(n-k+1), & \text{if \ $(\sum_{i=1}^{k} a_i, n) \nmid b$}; \\
(-1)^{k-1}(k-1)!\left((\sum_{i=1}^{k} a_i, n)-1\right)+(n-1)\cdots(n-k+1), & \text{if \ $(\sum_{i=1}^{k} a_i, n) \mid b$}.
\end{dcases}
\end{align*}
\end{theorem}

\begin{proof}
Let $\langle x_1,\ldots,x_k \rangle \in \Z_{n}^k$ be a solution of the linear congruence $a_1x_1+\cdots +a_kx_k\equiv b \pmod{n}$. Note that our desired solutions are those for which none of the $\binom{k}{2}$ equalities $x_u = x_v$, $1\leq u<v \leq k$, holds. Let $T_k=\lbrace \lbrace u,v \rbrace : 1\leq u<v \leq k \rbrace$. By the inclusion-exclusion principle, the number of such solutions is  
\begin{align} \label{iep}
N_n(b;a_1,\ldots,a_k)&=\sum_{e=0}^{(_{2}^{k})}(-1)^{e}\sum_{\substack{S \subseteq T_k \\ |S|=e}}N(S),
\end{align}
where $N(S)$ is the number of solutions of the linear congruence with $x_{\alpha} = x_{\beta}$ for $\lbrace \alpha,\beta \rbrace \in S$. 

Now, we need to calculate 
$$\sum_{\substack{S \subseteq T_k \\ |S|=e}}N(S).$$ 
In order to calculate $N(S)$, we construct the graph $G(S)$ on vertices $1,\ldots,k$ and edge set $S$. In calculating $N(S)$ we note that all vertices $i$ in a connected component of $G(S)$ correspond to the same $x_i$ in the linear congruence (by the definition of $N(S)$), and so we can simplify the linear congruence by grouping the $x_i$ which are equal to each other. This procedure eventually gives a new linear congruence in which the coefficients are of the form $\sum_{i \in I} a_i$, where $\emptyset \not= I\subseteq \lbrace 1, \ldots, k\rbrace$, and the number of terms is equal to the number of connected components of $G(S)$. If $G(S)$ has $c>1$ connected components then since $(\sum_{i \in I} a_i, n)=1$ for all $\emptyset \not= I\varsubsetneq \lbrace 1, \ldots, k\rbrace$, by Proposition~\ref{Prop: lin cong} we have $N(S)=n^{c-1}$. Also, if $G(S)$ is connected, that is, $c=1$ then $N(S)$ is the number of solutions of the linear congruence $(\sum_{i=1}^{k}a_i)x\equiv b \pmod{n}$, and so by Proposition~\ref{Prop: lin cong}, $N(S)$ in this case, we denote it by $A$, is equal to $(\sum_{i=1}^{k} a_i, n)$ if $(\sum_{i=1}^{k} a_i, n) \mid b$, and is equal to zero otherwise. Let $g(c,e,k)$ be the number of simple graphs with $c$ connected components, $e$ edges, and $k$ vertices labeled $1,\ldots,k$, and $g'(e,k)$ be the number of simple {\it connected} graphs with $e$ edges and $k$ labeled vertices. Now, recalling (\ref{iep}), we get
\begin{align*}
N_n(b;a_1,\ldots,a_k)&=\sum_{e=0}^{(_{2}^{k})}(-1)^{e}\left(Ag'(e,k)+\sum_{c=2}^{k}n^{c-1}g(c,e,k)\right)
\\
&=A\sum_{e=0}^{(_{2}^{k})}(-1)^{e}g'(e,k)+\frac{1}{n}\sum_{e=0}^{(_{2}^{k})}\sum_{c=2}^{k}(-1)^{e}n^{c}g(c,e,k)\\
&=(A-1)\sum_{e=0}^{(_{2}^{k})}(-1)^{e}g'(e,k)+\frac{1}{n}\sum_{e=0}^{(_{2}^{k})}\sum_{c=1}^{k}(-1)^{e}n^{c}g(c,e,k).
\end{align*}
Now, in order to evaluate the latter expression, we use the two formulas mentioned in Theorem~\ref{graph enum}. In fact, by Theorem~\ref{graph enum}, we have 

$$
\sum_{e,k}(-1)^{e}g'(e,k)\frac{z^k}{k!}=\log F(z,0),
$$
and

$$
\sum_{c,e,k}(-1)^{e}n^{c}g(c,e,k)\frac{z^k}{k!}=F(z,0)^n,
$$
where $F$ is the deformed exponential function. Note that $F(z,0)=1+z$. Now, we have 

$$
\sum_{e=0}^{(_{2}^{k})}(-1)^{e}g'(e,k)=\text{the coefficient of $\frac{z^k}{k!}$ in $\log(1+z)$, which is equal to $\frac{k!(-1)^{k+1}}{k}$},
$$
and

$$
\sum_{e=0}^{(_{2}^{k})}\sum_{c=1}^{k}(-1)^{e}n^{c}g(c,e,k)=\text{the coefficient of $\frac{z^k}{k!}$ in $(1+z)^{n}$, which is equal to $k!(_{k}^{n})$}.
$$

Consequently, the number $N_n(b;a_1,\ldots,a_k)$ of solutions $\langle x_1,\ldots,x_k \rangle \in \Z_{n}^k$ of the linear congruence $a_1x_1+\cdots +a_kx_k\equiv b \pmod{n}$, with all $x_i$ distinct, is 
\begin{align*}
& N_n(b;a_1,\ldots,a_k)=\frac{(A-1)k!(-1)^{k+1}}{k}+\frac{k!(_{k}^{n})}{n}\\
&=\begin{dcases}
(-1)^{k}(k-1)!+(n-1)\cdots(n-k+1), & \text{if \ $(\sum_{i=1}^{k} a_i, n) \nmid b$}; \\
(-1)^{k-1}(k-1)!\left((\sum_{i=1}^{k} a_i, n)-1\right)+(n-1)\cdots(n-k+1), & \text{if \ $(\sum_{i=1}^{k} a_i, n) \mid b$}.
\end{dcases}
\end{align*}
\end{proof}

\begin{rema}
Note that in Sch\"{o}nemann's theorem, $b$ is zero and $n$ is prime but in Theorem~\ref{Gener Schonemann thm}, both $b$ and $n$ are arbitrary.
\end{rema}

It would be an interesting problem to see if the technique presented in this paper can be modified so that it covers the problem in its full generality. So, we pose the following question.

\bigskip

\noindent{\textbf{Problem 1.}} Let $a_1,\ldots,a_k,b,n$ ($n\geq 1$) be arbitrary integers. Give an explicit formula for the number of solutions $\langle x_1,\ldots,x_k \rangle \in \Z_{n}^k$ of the linear congruence $a_1x_1+\cdots +a_kx_k\equiv b \pmod{n}$ with all $x_i$ distinct.

Such results would be interesting from several aspects. As we mentioned in the Introduction, the number of solutions of the linear congruence with the restrictions $(x_i,n)=t_i$ ($1\leq i\leq k$), where $t_1,\ldots,t_k$ are given positive divisors of $n$, has found very interesting applications in number theory, combinatorics, geometry, computer science, cryptography etc. Therefore, having an explicit formula for the number of solutions with all $x_i$ distinct may also lead to interesting applications in these or other directions. The problem may also have implications in zero-sum theory (see \cite{ADP, GPP}) and 
in coding theory (see \cite{BKS7}).

\section*{Acknowledgements}

The authors are grateful to the anonymous referees for a careful reading of the paper and helpful comments.

\end{document}